\documentclass[a4paper,UKenglish]{lipics-v2018}

\usepackage{microtype}

\title{Multistage Knapsack}

\titlerunning{Multistage Knapsack}

\author{Evripidis Bampis}{Sorbonne Universit\'{e}, CNRS, LIP6, France}{evripidis.bampis@lip6.fr}{}{}

\author{Bruno Escoffier}{Sorbonne Universit\'{e}, CNRS, LIP6, France}{bruno.escoffier@lip6.fr}{}{}

\author{Alexandre Teiller}{Sorbonne Universit\'{e}, CNRS, LIP6, France}{alexandre.teiller@lip6.fr}{}{}

\authorrunning{E. Bampis, B. Escoffier and A. Teiller}

\Copyright{CC-BY}

\subjclass{}


\category{}

\relatedversion{}

\supplement{}

\funding{}

\acknowledgements{}

\EventEditors{}
\EventLongTitle{}
\EventShortTitle{}
\EventAcronym{}
\EventYear{}
\EventDate{}
\EventLocation{}
\EventLogo{}
\SeriesVolume{}
\ArticleNo{}
\hideLIPIcs  

\nolinenumbers

\newtheorem{defn}{Definition}
\newtheorem{prop}{Proposition}

\begin{document}

\maketitle

\begin{abstract}
Many systems have to be maintained while the underlying constraints, costs and/or profits change over time. Although the state of a system may evolve during time, a non-negligible transition cost is incured for transitioning from one state to another. In order to model such situations, Gupta et al. (ICALP 2014) and Eisenstat et al. (ICALP 2014) introduced a \emph{multistage} model  where the input is a sequence of instances (one for each time step), and the goal is to find a sequence of solutions (one for each time step) that are both (i) near optimal  for each time step and (ii) as stable as possible. 
We focus on the \emph{multistage} version of the {\sc Knapsack} problem where we are 
given a time horizon $t=1,2,\ldots,T$, and a sequence of knapsack instances $I_1,I_2,\ldots,I_T$, one for each time step, defined on a set of $n$ objects. In every time step $t$ we have to choose a feasible knapsack $S_t$ of $I_t$, which gives a \emph{knapsack profit}. To measure the stability/similarity of two  consecutive solutions  $S_t$ and $S_{t+1}$, we identify the objects for which the decision, to be picked or not, remains the same in $S_t$ and $S_{t+1}$, giving a \emph{transition profit}.
We are asked to produce a sequence of solutions $S_1,S_2,\ldots,S_T$  so that the total knapsack profit plus the overall transition profit is maximized.

We propose a PTAS for the {\sc Multistage  Knapsack} problem. This is the first approximation scheme for a combinatorial optimization problem in the considered multistage setting, and its existence contrasts with the inapproximability results for other combinatorial optimization problems that are even polynomial-time solvable in the static case (e.g.{\sc multistage  Spanning Tree}, or {\sc multistage Bipartite Perfect Matching}). Then, we prove that there is no FPTAS for the problem even in the case where $T=2$, unless $P=NP$. Furthermore, we give a pseudopolynomial time algorithm for the case where the number of steps is bounded by a fixed constant and we show that otherwise the problem remains NP-hard even in the case where all the weights, profits and capacities are 0 or 1. 
 \end{abstract}


\section{Introduction}
In a classical combinatorial optimization problem, given an instance of the problem we seek a feasible solution optimizing the objective function. However, in many systems the input may change over the time and the solution has to be adapted to the input changes. It is then necessary to determine a tradeoff between the optimality of the solutions in each time step and the stability/similarity of consecutive solutions. This is important since in many applications there is a significant transition cost for changing (parts of) a solution. 
Recently, Gupta et al. \cite{Gupta} and Eisenstat et al. \cite{Eisenstat} introduced 
 a \emph{multistage} model in order to deal with such situations. They consider that the input is a sequence of instances (one for each time step), and the goal is to find a sequence of solutions (one for each time step) reaching such a tradeoff.

Our work follows the direction proposed by Gupta et al. \cite{Gupta} who suggested the study
of more combinatorial optimization problems in their multistage framework. In this paper,
we focus on the multistage version of the {\sc Knapsack} problem. Consider a company owning
a set $N = \{u_1, \ldots , u_n\}$ of production units. Each unit can be used or not; if $u_i$ is used, it
spends an amount $w_i$ of a given resource (energy, raw material,...), and generates a profit
$p_i$. Given a bound $W$ on the global amount of available resource, the static {\sc Knapsack} problem
aims at determining a feasible solution that specifies the chosen units in order to maximize the total profit under the constraint that the total amount of the resource does not exceed the bound of $W$.
In a multistage setting, considering a time horizon $t = 1, 2, \ldots,T$ of, let us say, $T$ days,
the company needs to decide a production plan for each day of the time horizon, given that data
(such as prices, level of resources,...) usually change over time. This is a typical situation,
for instance, in energy production planning (like electricity production, where units can be
nuclear reactors, wind or water turbines,...), or in data centers (where units are machines and
the resource  corresponds to the available energy). 
Moreover, in these examples, there is an extra cost to turn ON or OFF
 a unit like in the case of turning ON/OFF a reactor in electricity production~\cite{thesececile}, or a machine in a data center~\cite{Albers17}. 
Obviously, whenever a reactor is
in the ON or OFF state, it is beneficial to maintain it at the same state for several consecutive time steps, in order to avoid the overhead costs of state changes. 
Therefore, the design of a production plan over a given time horizon has to take into account 
both the profits generated each
day from the operation of the chosen units, as well as the potential transition profits from
maintaining a unit at the same state for consecutive days.
We refer the reader interested in planning problems in electricity production to~\cite{thesececile}.

We formalize the problem as follows. We are 
given a time horizon $t=1,2,\ldots,T$, and a sequence of knapsack instances $I_1,I_2,\ldots,I_T$, one for each time step, defined on a set of $n$ objects. In every time step $t$ we have to choose a feasible knapsack $S_t$ of $I_t$, which gives a \emph{knapsack profit}. Taking into account transition costs, we measure the stability/similarity of two  consecutive solutions  $S_t$ and $S_{t+1}$ by identifying the objects for which the decision, to be picked or not, remains the same in $S_t$ and $S_{t+1}$, giving a \emph{transition profit}. We are asked to produce a sequence of solutions $S_1,S_2,\ldots,S_T$  so that the total knapsack profit plus the overall transition profit is maximized.

Our main contribution is a polynomial time approximation scheme (PTAS) for the multistage  version of the {\sc Knapsack} problem. Up to the best of our knowledge, this is the first approximation scheme for a multistage combinatorial optimization problem and its existence contrasts with the inapproximability results for other combinatorial optimization problems that are even polynomial-time solvable in the static case (e.g. the {\sc multistage  Spanning Tree} problem \cite{Gupta}, or the {\sc multistage  Bipartite Perfect Matching} problem \cite{Bampis}).

\subsection{Problem definition}
Formally, the {\sc Multistage  Knapsack} problem can be defined as follows.

\begin{defn} In the {\sc Multistage  Knapsack} problem ($MK$) we are given:
\begin{itemize}
\item a time horizon $T \in \mathbb{N}^*$, a set $N=\{1,2,\dots,n\}$ of objects;
\item For any $t \in \{1,\dots,T\}$, any $i\in N$:
\begin{itemize}
\item $p_{ti}$ the profit of taking object $i$ at time $t$
\item $w_{ti}$ the weight of object $i$ at time $t$
\end{itemize} 
\item For any $t \in \{1,\dots,T-1\}$, any $i\in N$:
\begin{itemize}
    \item $B_{ti}  \in \mathbb{R^{+}}$  the bonus of the object $i$ if we keep the same decision for $i$ at time $t$ and $t+1$. 
\end{itemize}
\item For any $t \in \{1,\dots,T\}$: the capacity $C_t$ of the knapsack at time $t$.
\end{itemize}
We are asked to select a subset $S_t\subseteq N$ of objects at each time $t$ so as to respect the capacity constraint: $\sum_{i\in S_t} w_{ti}\leq C_t$. To a solution $S=(S_1,\dots,S_T)$ are associated:
\begin{itemize}
  \item A knapsack profit $\sum_{t=1}^T\sum_{i\in S_t} p_{ti}$ corresponding to the sum of the profits of the $T$ knapsacks;
  \item A transition profit $\sum\limits_{t=1} ^ {T-1} \sum\limits_{i \in \Delta_t} B_{ti}$ where $\Delta_t$ is the set of objects either taken or not taken at both time steps $t$ and $t+1$ in $S$ (formally $\Delta_t=(S_t\cap S_{t+1})\cup(\overline{S_t}\cap \overline{S_{t+1}})$).\\
  
\end{itemize} 
The value of the solution $S$ is the sum of the knapsack profit and the transition profit, to be maximized.
\end{defn}

\subsection{Related works}

\noindent{\bf Multistage combinatorial optimization.}  A lot of optimization
problems have been considered in online or semi-online settings, where the input changes over
time and the algorithm has to modify the solution (re-optimize) by making as few changes
as possible. We refer the reader to \cite{Anthony, Blanchard, Cohen, Gu, Megow, Nagarajan}
 and the references therein. 
 
Multistage optimization has been studied for fractional problems by Buchbinder et
al. \cite{Buchbinder} and Buchbinder, Chen and Naor \cite{Buchbinder+}. The multistage model considered in this article is the one studied in Eisenstat
et al. \cite{Eisenstat} and Gupta et al. \cite{Gupta}. Eisenstat
et al. \cite{Eisenstat} studied the multistage version of facility location problems. They proposed a logarithmic approximation algorithm. An et al. \cite{An} obtained constant factor approximation for some related problems.
 Gupta et al. \cite{Gupta} studied the {\sc Multistage  Maintenance Matroid} problem for
both the offline and the online settings. They presented a logarithmic approximation
algorithm for this problem, which includes as a special case a natural multistage version of
{\sc Spanning Tree}. The same paper also introduced the study of the {\sc Multistage  Minimum Perfect Matching} problem. They showed that the problem becomes hard to approximate
 even for a constant number of stages. Later, Bampis et al. \cite{Bampis} showed that the problem is
hard to approximate even for bipartite graphs and for the case of two time steps. In the case where the edge costs are metric within every time step they first proved that the problem remains APX-hard even for two time steps. They also show that the maximization version of the problem admits a constant factor approximation algorithm but is APX-hard. In another work
\cite{Bampis+}, the {\sc Multistage Max-Min Fair Allocation}
problem has been studied in the offline and  the online settings. This corresponds to a multistage variant of the {\sc Santa Klaus} problem. For the off-line setting, the authors
showed that the multistage version of the problem is much harder
than the static one. They provide
constant factor approximation algorithms for the off-line setting.\\

\noindent{\bf Knapsack variants.}

Our work builds upon the {\sc Knapsack} literature \cite{Kelerrer}. It is well known that there is a simple 2-approximation algorithm as well as a fully polynomial time (FPTAS) for the static case \cite{Ibarra, Lawler, Magazine, Kellerer}. There are two variants that are of special interest for our work: 

(i) The first variant is a  generalization of the {\sc Knapsack} problem known as the $k$-{\sc Dimensional Knapsack} ($k-DKP$) problem:
\begin{defn}
In the $k$-dimensional {\sc Knapsack} problem ($k-DKP$), we have a set $N=\{1,2,\dots,n\}$ of objects. Each object $i$ has a profit $p_i$ and $k$ weights $w_{ji}$, $j=1,\dots,k$. We are also given $k$ capacities $C_j$. The goal is to select a subset $Y\subseteq N$ of objects such that:
\begin{itemize}
  \item The capacity constraints are respected: for any $j$, $\sum_{i\in Y}w_{ji}\leq C_j$;
  \item The profit $\sum_{i\in Y} p_i$ is maximized. 
\end{itemize} 
\end{defn}

It is well known that for the usual {\sc Knapsack} problem, in the continuous relaxation (variables in $[0,1]$), at most one variable is fractional. Caprara et al. \cite{Carpara} showed that this can be generalized for $k-DKP$.

Let us consider the following ILP formulation $(ILP-DKP)$ of the problem:

   \begin{eqnarray*}
     \ \left \{ \begin{array}{ll}
    \max \ \sum\limits_{i \in N} p_{i} y_{i}   \\
    s.t. \left |
    \begin{array}{llllll}
    \sum\limits_{i \in N} w_{ji} y_{i} & \leq &C_j & \forall{j} \in \{1,...,k\} \\
    y_{i} \in \{0,1\} & & & \forall{i} \in N\\
    \end{array}
    \right.
    \end{array} 
    \right.
    \end{eqnarray*}

\begin{theorem}\cite{Carpara}\label{th:kdkp}
In the continuous relaxation $(LP-DKP)$ of $(ILP-DKP)$ where variables are in $[0,1]$, in any basic solution at most $k$ variables are fractional.
\end{theorem} 

Note that with an easy affine transformation on variables, the same result holds when variable $y_i$ is subject to $a_i\leq y_i\leq b_i$ instead of $0\leq y_i\leq 1$: {\it in any basic solution at most $k$ variables $y_i$ are such that $a_i<y_i<b_i$}.

Caprara et al. \cite{Carpara} use the result of Theorem~\ref{th:kdkp} to show that for any fixed constant $k$ $(k-DKP)$ admits a polynomial time approximation scheme (PTAS). Other PTASes
have been presented in  \cite{Oguz, Frieze}.  Korte and Schrader \cite{Korte} showed that there is no FPTAS for $k-DKP$ unless $P=NP$. 

(ii) The second related variant is a simplified version of $(k-DKP)$ called $\text{CARDINALITY}(2-KP)$, where the dimension is 2, all the profits are 1 and, given a $K$, we are asked if there is a solution of value at least $K$ (decision problem). In other words, given two knapsack constraints, can we take $K$ objects and verify the two constraints? The following result is shown in \cite{Kelerrer}. 

\begin{theorem}\cite{Kelerrer}\label{th:dimkp}
$\text{CARDINALITY}(2-KP)$ is $NP$-complete.
\end{theorem}

\subsection{Our contribution}
As stated before, our main contribution is to propose a PTAS for the {\sc multistage  Knapsack} problem. Furthermore, we prove that there is no  FPTAS for the problem even in the case where $T=2$, unless $P=NP$. We also give a pseudopolynomial time algorithm for the case where the number of steps is bounded by a fixed constant and we show that otherwise the problem remains NP-hard even in the case where all the weights, profits and capacities are 0 or 1. The following table summarizes our main result pointing out the impact of the number of time steps on the difficulty of the problem (``no FPTAS'' means ``no FPTAS unless P=NP'').

\begin{center}
\begin{tabular}{|l|c|r|}
  \hline
  $T=1$ & $T$ fixed &any $T$  \\
  \hline
  pseudopolynomial & pseudopolynomial & strongly $NP$-hard \\
    \hline
  FPTAS & PTAS & PTAS \\
   - & no FPTAS & no FPTAS \\
  \hline
\end{tabular}
\end{center}

We point out that the negative results (strongly NP-hardness and no FPTAS) hold even in the case of {\it uniform bonus} when $B_{ti}=B$ for all $i\in N$ and all $t=1,\dots,T-1$. 




\section{ILP formulation}

The {\sc Multistage  Knapsack} problem can be written as an ILP as follows. We define $Tn$ binary variables $x_{ti}$ equal to 1 if $i$ is taken at time $t$ ($i\in S_t$) and 0 otherwise. We also define $(T-1)n$ binary variables $z_{ti}$ corresponding to the transition profit of object $i$ between time $t$ and $t+1$. The profit is 1 if $i$ is taken at both time steps, or taken at none, and 0 otherwise. Hence, $z_{ti}=1-|x_{(t+1)i}-x_{ti}|$. Considering that we solve a maximization problem, this can be linearized by the two inequalities: $z_{ti}\leq -x_{(t+1)i} + x_{ti} +1$ and $z_{ti}\leq x_{(t+1)i} - x_{ti} +1$. We end up with the following ILP (called $ILP-MK$):  


\begin{center}
    \begin{eqnarray*}
     \ \left \{ \begin{array}{ll}
    \max \ \sum\limits_{t=1} ^ T \sum\limits_{i \in N} p_{ti} x_{ti} + \sum\limits_{t=1} ^ {T-1} \sum\limits_{i \in N} z_{ti}B_{ti}  \\
    s.t. \left |
    \begin{array}{llllll}
    \sum\limits_{i \in N} w_{ti} x_{ti} & \leq &C_t & \forall{t} \in \{1,...,T\} \\
    z_{ti} &\leq &-x_{(t+1)i} + x_{ti} +1 & \forall{t} \in \{1,...,T-1\},\forall{i} \in N\\
    z_{ti} &\leq & x_{(t+1)i} - x_{ti}+1 & \forall{t} \in \{1,...,T-1\},\forall{i} \in N\\
    x_{ti} \in \{0,1\} & & & \forall{t} \in \{1,...,T\}, \forall{i} \in N\\
    z_{ti} \in \{0,1\} & & & \forall{t} \in \{1,...,T-1\}, \forall{i} \in N\\
    \end{array}
    \right.
    \end{array} 
    \right.
    \end{eqnarray*}
    \end{center}

In devising the PTAS we will extensively use the linear relaxation $(LP-MK)$ of $(ILP-MK)$
where variables $x_{ti}$ and $z_{ti}$ are in $[0,1]$.


\section{A polynomial time approximation scheme}

In this section we show that {\sc Multistage  Knapsack} admits a PTAS. The central part of the proof is to derive a PTAS when the number of steps is a fixed constant (Sections~\ref{sec:frac} and \ref{sec:constant}). The generalization to an arbitrary number of steps is done in Section~\ref{sec:arbi}.

To get a PTAS for a constant number of steps, the proof follows the two main ideas leading to the PTAS for $k-DKP$ in \cite{Carpara}. Namely, for $k-DKP$:

\begin{itemize}
	\item The number of fractional variables in the continuous relaxation of $k-DKP$ is at most $k$ (Theorem~\ref{th:kdkp}); 
	\item A combination of bruteforce search (to find the most profitable objects) and LP-based solution allows to compute a solution close to optimal.
\end{itemize}

The main difficulty is to obtain a similar result for the number of fractional variables in the (relaxed) {\sc Multistage  Knapsack} problem: we end up with a result stating that there are at most $T^3$ fractional variables in an optimal (basic) solution. The brute force part is similar in essence though some additional difficulties are overcome by an additional preprocessing step. 

We show how to bound the number of fractional variables in Section~\ref{sec:frac}. We first illustrate the reasoning on the case of two time-steps, and then present the general result. In Section~\ref{sec:constant} we present the PTAS for a constant number of steps.  

For ease of notation, we will sometimes write a feasible solution as $S=(S_1,\dots,S_T)$ (subsets of objects taken at each time step), or as $S=(x,z)$ (values of variables in $(ILP-MK)$ or $(LP-MK)$).

\subsection{Bounding the number of fractional objects in $(LP-MK)$}\label{sec:frac}

\subsubsection{Warm-up: the case of two time-steps}

We consider in this section the case of two time-steps ($T=2$), and focus on the linear relaxation $(LP-MK)$ of $(ILP-MK)$ with the variables $x_{ti}$ and $z_i$ in $[0,1]$ (we write $z_i$ instead of $z_{1i}$ for readability). We say that an object is {\it fractional} in a solution $S$ if $x_{1i}$, $x_{2i}$ or $z_i$ is fractional. 

Let us consider a  (feasible) solution $\hat{S}=(\hat{x},\hat{z})$ of $(LP-MK)$, where $\hat{z}_{i}=1-|\hat{x}_{2i}-\hat{x}_{1i}|$ (variables $\hat{z}_{i}$ are set to their optimal value w.r.t. $\hat{x}$).  

We show the following.

\begin{prop}\label{lem:2steps}
	If $\hat{S}$ is a basic solution of $(LP-MK)$, at most 4 objects are fractional. 
\end{prop}

\begin{proof}
 First note that since we assume $\hat{z}_i=1-|\hat{x}_{1i}-\hat{x}_{2i}|$,  if $\hat{x}_{1i}$ and $\hat{x}_{2i}$  are both integers then $\hat{z}_i$ is an integer. So if an object $i$ is fractional either $\hat{x}_{1i}$ or $\hat{x}_{2i}$ is fractional. 
  
 Let us denote:
 \begin{itemize}
 	\item $L$ the set of objects $i$ such that  $\hat{x}_{1i}=\hat{x}_{2i}$.
 	\item $P=N\setminus L$ the set of objects $i$ such that $\hat{x}_{1i}\neq \hat{x}_{2i}$.\\
 \end{itemize} 
 
We first show Fact 1.\\

\noindent {\it Fact 1.} In $P$ there is at most one object $i$ with $\hat{x}_{1i}$ fractional.\\
 
Suppose that there are two such objects $i$ and $j$. Note that since $0<|\hat{x}_{1i}-\hat{x}_{2i}|<1$, $\hat{z}_i$  is fractional, and so is $\hat{z}_j$. Then, for a sufficiently small $\epsilon >0$, consider the solution $S_1$ obtained from $\hat{S}$ by transfering at time 1 an amount $\epsilon$ of weight from $i$ to $j$ (and adjusting consequently $z_i$ and $z_j$). Namely, in $S_1$:
\begin{itemize}
	\item $x^1_{1i}=\hat{x}_{1i}-\frac{\epsilon}{w_{1i}}$, $z^1_i=\hat{z}_i-d_i\frac{\epsilon}{w_{1i}}$, where $d_i=1$ if $\hat{x}_{2i}>\hat{x}_{1i}$ and $d_i=-1$ if $\hat{x}_{2i}<\hat{x}_{1i}$ (since $i$ is in $P$ $\hat{x}_{2i}\neq \hat{x}_{1i}$).
	\item $x^1_{1j}=\hat{x}_{1j}+\frac{\epsilon}{w_{1j}}$, $z^1_j=\hat{z}_i+d_j\frac{\epsilon}{w_{1j}}$, where $d_j=1$ if $\hat{x}_{2j}>\hat{x}_{1j}$ and $d_j=-1$ otherwise.
\end{itemize}
Note that (for $\epsilon$ sufficiently small) $S_1$ is feasible. Indeed (1) $\hat{x}_{1i},\hat{x}_{1j},\hat{z}_{i}$ and $\hat{z}_j$ are fractional (2) the weight of the knapsack at time 1 is the same in $S_1$ and in $\hat{S}$ (3) if $\hat{x}_{1i}$ increases by a small $\delta$, if $\hat{x}_{2i}>\hat{x}_{1i}$ then $|\hat{x}_{2i}-\hat{x}_{i1}|$ decreases by $\delta$ so $\hat{z}_i$ can increase by $\delta$ (so $d_i=1$), and if  $\hat{x}_{2i}<\hat{x}_{i1}$ then $\hat{z}_i$ has to decrease by $\delta$ (so $d_i=-1$),  and similarly for $\hat{x}_{1j}$. \\

Similarly, let us define $S_2$ obtained from $\hat{S}$ with the reverse transfer (from $j$ to $i$). In $S_2$:
\begin{itemize}
	\item $x^2_{1i}=\hat{x}_{1i}+\frac{\epsilon}{w_{1i}}$, $z^2_i=\hat{z}_i+d_i\frac{\epsilon}{w_{1i}}$
	\item $x^2_{1j}=\hat{x}_{1j}-\frac{\epsilon}{w_{1j}}$, $z^2_j=\hat{z}_i-d_j\frac{\epsilon}{w_{1j}}$
\end{itemize}
As previously, $S_2$ is feasible. Then $\hat{S}$ is clearly a convex combination of $S_1$ and $S_2$ (with coefficient 1/2), so not a basic solution, and Fact 1 is proven. 

In other words (and this interpretation will be important in the general case), for this case we can focus on variables at time one, and interpret \emph{locally} the problem as a (classical, unidimensional) fractional knapsack problem. By locally, we mean that if $\hat{x}_{1i}<\hat{x}_{2i}$ then $x_{1i}$ must be in $[0,\hat{x}_{2i}]$ (in $S^1$, $x^1_{1i}$ cannot be larger than $\hat{x}_{2i}$, otherwise the previous value of $z^1_{i}$  would be erroneous); similarly if $\hat{x}_{1i}>\hat{x}_{2i}$ then $x_{1i}$ must be in $[\hat{x}_{2i},1]$. The profit associated to object $i$ is $p_{1i}+d_iB_{1i}$ (if $x_{i1}$ increases/decreases by $\epsilon$, then the knapsack profit increases/decreases by $p_{1i}\epsilon$, and the transition profit increases/decreases by $\epsilon d_i B_{1i}$, as explained above). Then we have at most one fractional variable, as in any fractional knapsack problem.\\

In $P$ there is at most one object $i$ with $\hat{x}_{1i}$ fractional. Similarly there is at most one object $k$ with $\hat{x}_{2k}$ fractional. In $P$, for all but at most two objects, both $\hat{x}_{1i}$ and $\hat{x}_{2i}$, and thus $\hat{z}_i$, are integers.\\

Note that this argument would not hold for variables in $L$. Indeed if $\hat{x}_{1i}=\hat{x}_{2i}$, then $\hat{z}_i=1$, and the transition profit decreases in {\it both} cases: when $\hat{x}_{1i}$ increases by $\delta>0$ and when it decreases by $\delta$. So, we cannot express $\hat{S}$ as a convex combination of $S_1$ and $S_2$ as previously. \\

However, let us consider the following linear program $2-DKP$ obtained by fixing variables in $P$ to their values in $\hat{S}$, computing the remaining capacities $C'_t=C_t-\sum_{j\in P}w_{tj}\hat{x}_{tj}$, and ``imposing'' $x_{1i}=x_{2i}$: 

\vspace{-0.5cm}

\[ \left \{ 
\begin{tabular}{ccc}
$ \max \sum\limits_{i \in L} (p_{1i}+p_{2i})y_{i}+\sum\limits_{i \in L} B_{1i}$\\
$  \sum\limits_{i \in L} w_{1i}y_{i}$ & $ \leq $ & $C_1'$\\
$  \sum\limits_{i \in L} w_{2i}y_{i}$ & $ \leq $ & $C_2'$\\
$y_{i} \in [0,1] $&&$\forall i \in L$ \\
\end{tabular} 
\right.
\] 

Clearly, the restriction of $\hat{S}$ to variables in $L$ is a solution of $2-DKP$. Formally, let $\hat{S}_L=(\hat{y}_j, j\in L)$ defined as $\hat{y}_j=\hat{x}_{1j}$. $\hat{S}_L$ is feasible for $2-DKP$. Let us show that it is basic: suppose a contrario that $\hat{S}_L=\frac{S^1_L+S^2_L}{2}$, with $S^1_L=(y^1_i,i\in L)\neq S^2_L$ two feasible solutions of  $2-DKP$. Then consider the solution $S^1=(x^1,y^1)$ of $(LP-MK)$ defined as:
\begin{itemize}
	\item If $i\in L$ then $x^1_{1i}=x^1_{2i}=y^1_i$, and $z^{1}_{1i}=1=\hat{z}_{1i}$.
	\item Otherwise (for $i$ in $P$) $S^1$ is the same as $\hat{S}$. 
\end{itemize}
$S^1$ is clearly a feasible solution of {\sc Multistage  Knapsack}. If we do the same for $S^2_L$, we get a (different) feasible solution $S^2$, and $\hat{S}=\frac{S^1+S^2}{2}$, so $\hat{S}$ is not basic, a contradiction. 

By the result of \cite{Carpara}, $\hat{S}_L$ has at most 2 fractional variables. Then, in $L$, for all but at most 2 variables both $\hat{x}_{1i}$, $\hat{x}_{2i}$ and $\hat{z}_i$ are integers.   
\end{proof}

\subsubsection{General case}

The case of 2 time steps suggests to bound the number of fractional objects by considering 3 cases:
\begin{itemize}
	\item Objects with $\hat{x}_{1i}$ fractional and $\hat{x}_{1i}\neq \hat{x}_{2i}$. As explained in the proof of Proposition~\ref{lem:2steps}, this can be seen locally (as long as $x_{1i}$ does not reach $\hat{x}_{2i}$) as a knapsack problem from which we can conclude that there is at most 1 such fractional object.
	\item Similarly, objects with $\hat{x}_{2i}$ fractional and $\hat{x}_{1i}\neq \hat{x}_{2i}$.
	\item Objects with $\hat{x}_{1i}=\hat{x}_{2i}$ fractional. As explained in the proof of Proposition~\ref{lem:2steps}, this can be seen as a $2-DKP$ from which we can conclude that there are at most 2 such fractional objects.
\end{itemize}

For larger $T$, we may have different situations. Suppose for instance that we have 5 time steps, and a solution $(x,z)$ with an object $i$ such that: $x_{1i}<x_{2i}=x_{3i}=x_{4i}<x_{5i}$. So we have $x_{ti}$ fractional and constant for $t=2,3,4$, and different from $x_{1i}$ and $x_{5i}$. The idea is to say that we cannot have many objects like this (in a basic solution), by interpreting these objects on time steps $3,4,5$ as a basic optimal solution of a $3-DKP$ (locally, i.e. with a variable $y_i$ such that $x_{1i}\leq y_i\leq x_{5i}$).  

Then, roughly speaking, the idea is to show that for any pair of time steps $t_0\leq t_1$, we can bound the number of objects which are fractional and constant on this time interval $[t_0,t_1]$ (but not at time $t_0-1$ and $t_1+1$). Then a sum on all the possible choices of $(t_0,t_1)$ gives the global upper bound.

Let us state this rough idea formally. In all this section, we consider a (feasible) solution $\hat{S}=(\hat{x},\hat{z})$ of $(LP-MK)$, where $\hat{z}_{ti}=1-|\hat{x}_{(t+1)i}-\hat{x}_{ti}|$ (variables $\hat{z}_{ti}$ are set to their optimal value w.r.t. $\hat{x}$). 

In such a solution $\hat{S}=(\hat{x},\hat{z})$, let us define as previously an object as {\it fractional} if at least one variable $\hat{x}_{ti}$ or $\hat{z}_{ti}$ is fractional.  Our goal is to show the following result. 

\begin{theorem}\label{theo:frac}
If $\hat{S}=(\hat{x},\hat{z})$ is a basic solution of $(LP-MK)$, it has at most $T^3$ fractional objects. 
\end{theorem}

Before proving the theorem, let us introduce some definitions and show some lemmas. Let $t_0,t_1$ be two time steps with $1\leq t_0\leq t_1 \leq T$. 

\begin{defn}\label{def:f}
	The set $F(t_0,t_1)$ associated to $\hat{S}=(\hat{x},\hat{z})$ is the set of objects $i$ (called fractional w.r.t. $(t_0,t_1)$) such that 
	\begin{itemize}
		\item $0<\hat{x}_{t_0i}=\hat{x}_{(t_0+1)i}=\dots = \hat{x}_{t_1i}<1$;
		\item Either $t_0=1$ or $\hat{x}_{(t_0-1)i}\neq \hat{x}_{t_0i}$;
		\item Either $t_1=T$ or $\hat{x}_{(t_1+1)i}\neq \hat{x}_{t_1i}$;
	\end{itemize}
\end{defn}

In other words, we have $\hat{x}_{ti}$ fractional and constant on $[t_0,t_1]$, and $[t_0,t_1]$ is maximal w.r.t. this property. 

For $t_0\leq t\leq t_1$, we note $C'_t$ the remaining capacity of knapsack at time $t$ considering that variables outside $F(t_0,t_1)$ are fixed (to their value in $\hat{x}$):
$$C'_t=C_t-\sum_{i\not\in F(t_0,t_1)}w_{ti}\hat{x}_{ti}.$$

As previously, we will see $x_{t_0i},\dots,x_{t_1i}$ as a single variable $y_i$. We have to express the fact that this variable $y_i$ cannot ``cross'' the values $\hat{x}_{(t_0-1)i}$ (if $t_0>1$) and $\hat{x}_{(t_1+1)i}$ (if $t_1<T$), so that everything remains locally (in this range) linear. So we define the lower and upper bounds $a_i,b_i$ induced by Definition~\ref{def:f} as:
\begin{itemize}
	\item Initialize $a_i\leftarrow 0$. If $\hat{x}_{(t_0-1)i}< \hat{x}_{t_0i}$ then do $a_i\leftarrow \hat{x}_{(t_0-1)i}$. If $\hat{x}_{(t_1+1)i}< \hat{x}_{t_1i}$ then do $a_i\leftarrow \max(a_i,\hat{x}_{(t_1+1)i})$.
	\item Similarly, initialize $b_i\leftarrow 1$. If $\hat{x}_{(t_0-1)i}> \hat{x}_{t_0i}$ then do $b_i\leftarrow \hat{x}_{(t_0-1)i}$. If $\hat{x}_{(t_1+1)i}> \hat{x}_{t_1i}$ then do $b_i\leftarrow \min(b_i,\hat{x}_{(t_1+1)i})$.
\end{itemize}

Note that with this definition $a_i<\hat{x}_{t_0,i}<b_i$.
This allows us to define the polyhedron $P(t_0,t_1)$ as the set of $y=(y_i:i\in F(t_0,t_1))$ such that

\[  \left \{ 
\begin{tabular}{ccccc}
$\sum\limits_{i \in F(t_0,t_1)} w_{ti} y_{i}$ &$\leq$ & $C'_t$&$\forall{t} \in \{t_0,...,t_1\}$  \\
$a_i\leq y_{i} \leq b_i$&&&$\forall{i} \in F(t_0,t_1)$ & \\
\end{tabular} 
\right.
\]

\begin{defn}
The solution associated to $\hat{S}=(\hat{x},\hat{z})$ is $\hat{y}$ defined as $\hat{y}_i=\hat{x}_{t_0i}$ for $i\in F(t_0,t_1)$.  
\end{defn}

\begin{lemma}\label{lem:basic}
If $\hat{S}=(\hat{x},\hat{z})$ is a basic solution, then the solution $\hat{y}$ associated to $(\hat{x},\hat{z})$ is feasible of $P(t_0,t_1)$ and basic. 
\end{lemma}

\begin{proof}
	Since $(\hat{x},\hat{z})$ is feasible, then $\hat{y}$ respects the capacity constraints (remaining capacity), and $a_i<\hat{y}_i=\hat{x}_{t_0i}<b_i$ so $\hat{y}$ is feasible.
	
	Suppose now that $\hat{y}=\frac{y^1+y^2}{2}$ for two feasible solutions $y^1\neq y^2$ of $P(t_0,t_1)$. We associate to $y^1$ a feasible solution $S^1=(x^1,z^1)$ as follows. 
	
	We fix $x^1_{ti}=\hat{x}_i$ for $t\not\in[t_0,t_1]$, and $x^1_{ti}=y^1_i$ for $t\in [t_0,t_1]$. We fix variables $z^1_{it}$ to their maximal values, i.e. $z^1_{ti}=1-|x^1_{(t+1)i}-x^1_{ti}|$. This way, we get a feasible solution $(x^1,z^1)$. Note that:
\begin{itemize}
	\item $z^1_{ti}=\hat{z}_{ti}$ for $t\not \in [t_0-1,t_1]$, since coresponding variables $x$ are the same in $S^1$ and $\hat{S}$;
	\item $z^1_{ti}=1=\hat{z}_{ti}$ for $t \in [t_0,t_1-1]$, since variables $x$ are constant on the interval $[t_0,t_1]$.
	\end{itemize} 
	Then, for variables $z$, the only modifications between $z^1$ and $\hat{z}$ concerns the ``boundary'' variables $z^1_{ti}$ for $t=t_0-1$ and $t=t_1$. 
	
We build this way two solutions $S^1=(x^1,z^1)$ and $S^2=(x^2,z^2)$ of $(LP-MK)$ corresponding to $y^1$ and $y^2$. By construction, $S^1$ and $S^2$ are feasible. They are also different provided that $y^1$ and $y^2$ are different. It remains to prove that $\hat{S}$ is the half sum of  $S^1$ and $S^2$.

Let us first consider variables $x$: 
\begin{itemize}
\item if $t\not \in [t_0,t_1]$, $x^1_{ti}=x^2_{ti}=\hat{x}_{ti}$ so $\hat{x}_{ti}=\frac{x^1_{ti}+x^2_{ti}}{2}$. 
\item if $t \in [t_0,t_1]$, $x^1_{ti}=y^1_i$ and $x^2_{ti}=y^2_t$, so $\frac{x^1_{ti}+x^2_{ti}}{2}=\frac{y^1_{i}+y^2_{i}}{2}=\hat{y}_i=\hat{x}_{ti}$.
\end{itemize}

Now let us look at  variables $z$: first, for $t\not\in \{t_0-1,t_1\}$, $z^1_{ti}=z^2_{ti}=\hat{z}_{ti}$ so $\hat{z}_{ti}=\frac{z^1_{ti}+z^2_{ti}}{2}$. The last and main part concerns about the last 2 variables $z_{(t_0-1)i}$ (if $t_0>1$) and $z_{t_1i}$ (if $t_1<T$).   	

We have $z^1_{(t_0-1)i}=1-|x^1_{t_0i}-x^1_{(t_0-1)i}|=1-|x^1_{t_0i}-\hat{x}_{(t_0-1)i}|$ and $\hat{z}_{(t_0-1)i}=1-|\hat{x}_{t_0i}-\hat{x}_{(t_0-1)i}|$. The crucial point is to observe that thanks to the constraint $a_i\leq y_i\leq b_i$, and by definition of $a_i$ and $b_i$,  $x^1_{t_0,i}$, $x^2_{t_0,i}$ and $\hat{x}_{t_0,i}$ are either all greater than (or equal to) $\hat{x}_{(t_0-1)i}$, or all lower than (or equal to)  $\hat{x}_{(t_0-1)i}$. 

Suppose first that they are all greater than (or equal to) $\hat{x}_{(t_0-1)i}$. Then:
$$z^1_{(t_0-1)i}-\hat{z}_{(t_0-1)i}=|\hat{x}_{t_0,i}-\hat{x}_{t_0-1,i}|-|x^1_{t_0,i}-\hat{x}_{t_0-1,i}|=\hat{x}_{t_0i}-x^1_{t_0i}=\hat{y}_i-y^1_i$$
Similarly, $z^2_{(t_0-1)i}-\hat{z}_{(t_0-1)i}=\hat{y}_i-y^2_i$. So
	$$\frac{z^1_{(t_0-1)i}+z^2_{(t_0-1)i}}{2}=\frac{2\hat{z}_{(t_0-1)i}+2\hat{y}_i-y^1_{i}-y^2_{i}}{2}=\hat{z}_{(t_0-1)i}.$$

Now suppose that they are all lower than (or equal to) $\hat{x}_{t_0-1,i}$. Then:
$$z^1_{(t_0-1)i}-\hat{z}_{(t_0-1)i}=|\hat{x}_{t_0i}-\hat{x}_{(t_0-1)i}|-|x^1_{t_0i}-\hat{x}_{(t_0-1)i}|=x^1_{t_0i}-\hat{x}_{t_0i}=y^1_i-\hat{y}_i$$
Similarly, $z^2_{(t_0-1)i}-\hat{z}_{(t_0-1)i}=y^2_i-\hat{y}_i$. So
	$$\frac{z^1_{(t_0-1)i}+z^2_{(t_0-1)i}}{2}=\frac{2\hat{z}_{(t_0-1)i}-2\hat{y}_i+y^1_{i}+y^2_{i}}{2}=\hat{z}_{(t_0-1)i}.$$

Then, in both cases, $\hat{z}_{(t_0-1)i}=\frac{z^1_{(t_0-1)i}+z^2_{(t_0-1)i}}{2}$.\\
	
With the very same arguments we can show that  $\frac{z^1_{t_1i}+z^2_{t_1i}}{2}=\hat{z}_{t_1i}$.
	
Then, $\hat{S}$ is the half sum of $S^1$ and $S^2$, contradiction with the fact that $\hat{S}$ is basic.
\end{proof}

Now we can bound the number of fractional objects w.r.t. $(t_0,t_1)$.

\begin{lemma}\label{lem:frac}
$|F(t_0,t_1)|\leq t_1+1-t_0$.
\end{lemma}
\begin{proof}
$P(t_0,t_1)$ is a polyhedron corresponding to a linear relaxation of a $k-DLP$, with $k=t_1+1-t_0$. Since $\hat{y}$ is basic, using Theorem~\ref{th:kdkp} (and the note after) there are at most $k=t_1+1-t_0$ variables  $\hat{y}_i$ such that  $a_i<\hat{y}_i<b_i$. But by definition of  $F(t_0,t_1)$, for all $i\in F(t_0,t_1)$ $a_i<\hat{y}_i<b_i$. Then $|F(t_0,t_1)|\leq  t_1+1-t_0$.
\end{proof}

Now we can easily prove Theorem~\ref{theo:frac}.

\begin{proof}
First note that if $\hat{x}_{ti}$ and $\hat{x}_{(t+1)i}$ are integral, then so is $\hat{z}_{ti}$. Then, if an object $i$ is fractional at least one $\hat{x}_{ti}$ is fractional, and so $i$ will appear in (at least) one set $F(t_0,t_1)$. 

We consider all pairs $(t_0,t_1)$ with $1\leq t_0\leq t_1 \leq T$. Thanks to Lemma~\ref{lem:frac}, $|F(t_0,t_1)|\leq t_1+1-t_0$. So, the total number of fractional objects is at most:

$$N_T=\sum_{t_0=1}^{T}\sum_{t_1=t_0}^T(t_1+1-t_0) \leq T^3$$

Indeed, there are less than $T^2$ choices for $(t_0,t_1)$ and at most $T$ fractional objects for each choice. 
\end{proof}

Note that with standard calculation we get $N_T=\frac{T^3+3T^2+2T}{6}$, so for $T=2$ time steps $N_2=4$: we have at most 4 fractional objects, the same bound as in Proposition~\ref{lem:2steps}.

\subsection{A PTAS for a constant number of time steps}\label{sec:constant}

Now we can describe the $PTAS$. Informally, the algorithm first guesses the $\ell$ objects with the maximum reward in an optimal solution (where $\ell$ is defined as a function of $\epsilon$ and $T$), and then finds a solution on the remaining instance using the relaxation of the LP. The fact that the number of fractional objects is small allows to bound the error made by the algorithm. 

For a solution $S$ (either fractional or integral) we define $g_i(S)$ as the reward of object $i$ in solution $S$: $g_i(S)=\sum_{t=1}^Tp_{ti}x_{ti}+\sum_{t=1}^{T-1}z_{ti}B_{ti}$. The value of a solution $S$ is $g(S)= \sum_{i\in N}g_i(S)$.


Consider the algorithm $A^{LP}$ which, on an instance $(ILP-MK)$ of {\sc Multistage Knapsack}:
\begin{itemize}
	\item Finds an optimal (basic) solution $S^r=(x^r,z^r)$ of the relaxation $(LP-MK)$ of $(ILP-MK)$;
	\item Takes at step $t$ an object $i$ if and only if $x^r_{ti}=1$.
\end{itemize}
Clearly, $A^{LP}$ outputs a feasible solution, the value of which verifies: 

\begin{equation}\label{eqlpr}
g(A^{LP})\geq g(S^r)-\sum_{i\in F} g_i(S^r)
\end{equation}

where $F$ is the set of fractional objects in $S^r$. Indeed, for each integral (i.e., not fractional) object the reward is the same in both solutions. \\

Now we can describe the algorithm \textbf{Algorithm $PTAS_{ConstantMK}$}, which takes as input an instance of {\sc Multistage  Knapsack} and an $\epsilon>0$.\\

\textbf{Algorithm $PTAS_{ConstantMK}$}
\begin{enumerate}
  \item Let $\ell:=\min\left\{\left\lceil \frac{(T+1)T^3}{\epsilon}\right\rceil,n\right\}$.
  \item For all $X \subseteq N$ such that $|X| =\ell$, $\forall X_1\subseteq X, ...,  \forall X_T\subseteq X$: 
 
   If for all $t=1,\dots,T$ $w_t(X_t) = \sum\limits_{j \in X_t} w_{tj} \leq C_t$, then:
   \begin{itemize}
   \item Compute the rewards of object $i\in X$ in the solution $(X_1,\dots,X_T)$, and find the smallest one, say $k$, with reward $g_k$. 
   \item On the subinstance of objects $Y=N\setminus X$: 
   	\begin{itemize}
   	 \item For all $i \in Y$, for all $t \in \{1,\dots,T\}$: if $p_{ti} > g_{k}$ then set $x_{ti} =0$.
     \item apply $A^{LP}$ on the subinstance of objects $Y$, with the remaining capacity $C'_t=C_t-w_t(X_t)$, where some variables $x_{ti}$ are set to 0 as explained in the previous step.
     \end{itemize}
  \item Let $(Y_1,..., Y_T)$ be the sets of objects taken at time $1,\dots,T$ by  $A^{LP}$. Consider the solution $(X_1\cup Y_1,...,X_T\cup Y_T)$.
  \end{itemize}
  \item Output the best solution computed.
\end{enumerate}

\begin{theorem}\label{theo:ptasconstant}
The algorithm $PTAS_{ConstantMK}$ is a $(1-\epsilon)$-approximation algorithm running in time $O\left(n^{O(T^5/\epsilon)}\right)$.
\end{theorem}

\begin{proof}
First, there are $O(n^\ell)$ choices for $X$; for each $X$ there are $2^\ell$ choices for each $X_t$, so in all there are $O(n^\ell 2^{\ell T})$ choices for $(X_1,\dots,X_T)$. For each choice $(X_1,\dots,X_t)$, we compute the reward of elements, and then apply $A^{LP}$. Since $\ell\leq \lceil T(T+1)^3/\epsilon\rceil$, the running time follows. 

Now let us show the claimed approximation ratio. Consider an optimal solution $S^*$, and suppose wlog that $g_i(S^*)$ are in non increasing order. Consider the iteration of the algorithm where $X=\{1,2,\dots,\ell\}$. At this iteration, consider the choice $(X_1,\dots,X_T)$ where $X_t$ is exactly the subset of objects in $X$ taken by $S^*$ at time $t$ (for $t=1,\dots,T$). The solution $S$ computed by the algorithm at this iteration (with $X$ and $(X_1,\dots,X_T)$) is $S=(X_1\cup Y_1,\dots,X_T\cup Y_T)$ where $(Y_1,\dots,Y_T)$ is the solution output by $A^{LP}$ on the subinstance of objects $Y=N\setminus X$, where $x_{ti}$ is set to 0 if $p_{ti}>g_k$.\\

Note that since we consider the iteration where $X$ corresponds to the $\ell$ objects of largest reward in the optimal solution, we do know that for an object $i>\ell$, if $p_{ti} > g_{k}$ then the optimal solution does not take object $i$ at time $t$ (it would have a reward greater than $g_k$), so we can safely fix this variable to 0. The idea behind putting these variables $x_{ti}$ to $0$ is to prevent the relaxed solution to take fractional objects with very large profits. These objects could indeed induce a very high loss when rounding the fractional solution to an integral one as done by $A^{LP}$. \\
By doing this, the number of fractional objects (i.e., objects in $F$) does not increase. Indeed, if we put a variable $x_{ti}$ at $0$, it is not fractional so nothing changes in the proof of Theorem~\ref{theo:frac}.


The value of $S^*$ is $g(S^*)=\sum_{i\in X} g_i(S^*)+\sum_{i\in Y} g_i(S^*)$. Thus, by equation~\ref{eqlpr}, we have:

\begin{eqnarray*}
g(S) & \geq & \sum_{i\in X} g_i(S^*)+\sum_{i\in Y} g_i(S^r)-\sum_{i\in F} g_i(S^r)\\
& \geq & \sum_{i\in X} g_i(S^*)+\sum_{i\in Y} g_i(S^*)-\sum_{i\in F} g_i(S^r),
\end{eqnarray*}

\noindent where $F$ is the set of fractional objects in the optimal fractional solution of the relaxation of the LP on $Y$, where $x_{ti}$ is set to 0 if $p_{ti}>g_k$.   \\

Each object of $F$ has a profit at most $g_k$ at any time steps and a transition profit at most $\sum_{t=1}^{T-1} B_{ti}$, so 
$$\forall i \in F \  g_i(S^r) \leq Tg_k + \sum_{t=1}^{T-1} B_{ti}$$ 
Now, note that in an optimal solution each object has a reward at least $\sum_{t=1}^{T-1} B_{ti}$ (otherwise simply never take this object), so for all $i \in F$  $g_k\geq g_i\geq  \sum_{t=1}^{T-1} B_{ti}$. So we have : $$\forall i \in F \  g_i(S^r) \leq  (T+1) g_k$$

Since $g_i(S^*)$ are in non increasing order, we have $g_k=g_{\ell}(S^*)  \leq \frac{\sum_{i\in X}g_i(S^*)}{\ell}$. So for all objects $i$ in $F$, $g_i(S^r) \leq \frac{(T+1)\sum_{i\in X}g_i(S^*)}{\ell}$. By Theorem~\ref{theo:frac}, there are at most $T^3$ of them, thus : 
\begin{eqnarray*}
g(S) & \geq &\sum_{i\in X} g_i(S^*)+\sum_{i\in Y} g_i(S^*)-\sum_{i\in F} g_i(S^r)\\
& \geq &\sum_{i\in X} g_i(S^*)+\sum_{i\in Y} g_i(S^*)-\frac{(T+1)T^3}{\ell}\sum_{i\in X}g_i(S^*)\\
& \geq & (1-\epsilon)\sum_{i\in X} g_i(S^*)+\sum_{i\in Y} g_i(S^*) \geq (1-\epsilon)g(S^*)
\end{eqnarray*}
\end{proof}

By Theorem~\ref{theo:ptasconstant}, for any fixed number of time steps $T$, {\sc Multistage  Knapsack} admits a PTAS.


\subsection{Generalization to an arbitrary number of time steps}\label{sec:arbi}




We now devise a PTAS for the general problem, for an arbitrary (not constant) number
 of steps. We actually show how to get such a PTAS provided that we have a PTAS for (any)
constant number of time steps.
Let $A_{\epsilon,T_0}$ be an algorithm which, given an instance of {\sc Multistage Knapsack} with at
most $T_0$ time steps, outputs a $(1-\epsilon)$-approximate solution in time $O(n^{f(\epsilon,T_0)})$
 for some
function $f$.

 The underlying idea is to compute (nearly) optimal solutions on subinstances of bounded
 sizes, and then to combine them in such a way that at most a small fraction of the optimal
value is lost. 

 Let us first give a rough idea of our algorithm $PTAS_{MK}$.

Given an $\epsilon >0$, let $\epsilon'=\epsilon/2$ and $T_0 =\lceil \frac{1}{\epsilon'}\rceil$. We construct a set of solutions $S^1,\ldots ,S^{T_0}$ in the following way: 

In order to construct $S^1$, we partition the time horizon $1,\ldots, T$ into  $\lceil \frac{T}{T_0}\rceil$ consecutive intervals. Every such interval has length $T_0$, except possibly the last interval that may have a smaller length. We apply $A_{\epsilon,T_0}$ at every interval in this partition. $S^1$ is  then just the  concatenation of the partial solutions computed for each interval.

The partition on which it is based the construction of the solution $S^i$, $1<i \leq T_0$, is made in a similar way. The only difference is that the first interval of the partition of the time horizon $1,\ldots, T$ goes from time 1 to time $i-1$. For the remaining part of the time horizon, i.e. for $i, \ldots T$,  the partition is made as previously, i.e. starting at time step $i$, every interval will have a length of $T_0$, except possibly the last one, whose length may be smaller. Once the partition is operated,  we apply $A_{\epsilon,T_0}$ to every interval of the partition. $S^i$, $1<i \leq T_0$, is then defined as the concatenation of the partial solutions computed on each interval.
Among the $T_0$ solutions  $S^1,\ldots ,S^{T_0}$, the algorithm chooses the best solution.


The construction is illustrated on Figure 1, with 10 time steps and $T_0 = 3$. The first
solution $S^1$ is built by applying 4 times $A_{\epsilon,T_0}$, on the subinstances corresponding to time
steps $\{1, 2, 3\}, \{4, 5, 6\}, \{7, 8, 9\}$, and $\{10\}$. The 
solution $S^2$ is built by applying 4 times $A_{\epsilon,T_0}$, on the subinstances corresponding to time
steps $\{1\}, \{2, 3, 4\}, \{5, 6, 7\}$, and $\{8,9, 10\}$.

\begin{figure}[ht]
\begin{center}
\includegraphics[scale=0.5]{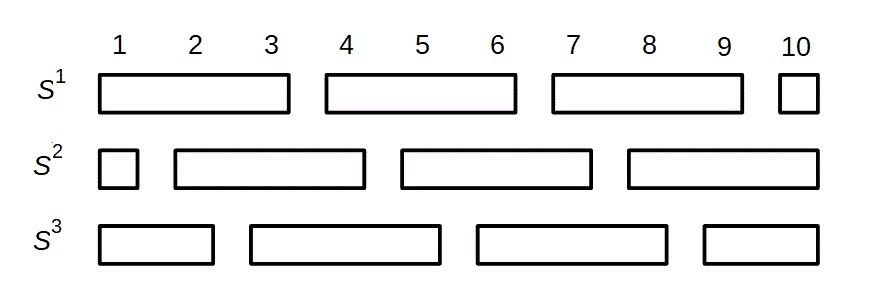}
\caption{The three solutions for $T_0=3$ and $T=10$.}\label{fig:solutions}
\end{center}
\end{figure}

More formally, given an $\epsilon>0$, the algorithm $PTAS_{MK}$ works as follows.
\begin{itemize}
\item Let $\epsilon'=\epsilon/2$ and $T_0=\left\lceil \frac{1}{\epsilon'}\right\rceil$. Let ${\cal I}_t=[t,\dots, t+T_0-1] \cap [1, \dots, T]$ be the  set of (at most) $T_0$ consecutive time steps starting at $t$. We consider ${\cal I}_t$ for any $t$ for which it is non empty (so $t\in [-T_0+2 \dots T]$).
\item For $t \in \{1, \ldots, T_0 \}$:
\begin{itemize}
    \item Apply $A_{\epsilon',T_0}$ on all intervals ${\cal I}_{t'}$ with $t'\equiv t \mod T_0$.
    Note that each time step belongs to exactly one of such intervals. 
    \item Define the solution $S^t$ built from the partial solutions given by the applications of $A_{\epsilon'\textbf{},T_0}$.
\end{itemize}
\item Choose the best solution $S$ among the $T_0$ solutions $S^1,\dots,S^{T_0}$.
\end{itemize} 


\begin{theorem}\label{thptasgen}
The algorithm $PTAS_{MK}$ is a polynomial time approximation algorithm.
\end{theorem}

\begin{proof}
The algorithm calls the $A_{\epsilon',T_0}$ algorithm $\lceil T/T_0 \rceil$ times for each of the $T_0$ generated solutions. Yet, the running time of  $A_{\epsilon',T_0}$ is  $n^{f(\frac{1}{\epsilon'},T_0)}=n^{f(\frac{2}{\epsilon},\lceil \frac{2}{\epsilon}\rceil)}$, i.e a polynomial time for any fixed $\epsilon$.  So, the running time of the algorithm for any $T$ is\\ $T_0\times \left \lceil \frac{T}{T_0} \right \rceil n^{f(\frac{2}{\epsilon},\frac{2}{\epsilon})} = O(T n^{f(\frac{2}{\epsilon},\frac{2}{\epsilon})})$, a polynomial time for any fixed $\epsilon$. \\

Each solution $S^t$ of the $T_0$ generated solutions may loose some bonus between the last time step of one of its intervals ${\cal I}_{t+kT_0}$ and the first time step of its next interval ${\cal I}_{t+(k+1)T_0}$ (in Figure~\ref{fig:solutions} for instance, $S^1$ misses the bonuses between steps 3 and 4, 6 and 7, and 9 and 10). Let $loss(S^t)$ be this loss with respect to some optimal solution $S^*$. Since we apply $A_{\epsilon',T_0}$ to build solutions, we get that the value $g(S^t)$ of $S^t$ is such that:

$$g(S^t)\geq (1-\epsilon') g(S^*) - loss(S^t)$$

Since the output solution $S$ is the best among the solutions $S^t$, by summing up the previous inequality for $t=1,\dots,T_0$ we get:

$$g(S)\geq (1-\epsilon') g(S^*)-\frac{\sum_{t=1}^T loss(S^t)}{T_0}$$

Now, by construction the bonus between steps $j$ and $j+1$ appears in the loss of  exactly one $S^t$ (see Figure~\ref{fig:solutions}). So the total loss of the $T_0$ solutions is the global transition bonus of $S^*$, so at most $g(S^*)$. Hence:

\begin{eqnarray*}
g(S) &\geq &  (1-\epsilon')g(S^*)-\frac{g(S^*)}{T_0} \\
&\geq&(1-2\epsilon')g(S^*) = (1-\epsilon) g(S^*)
\end{eqnarray*}

\end{proof}


\section{Pseudo-polynomiality and hardness results}

We complement the previous result on approximation scheme by showing the following results for {\sc Multistage  Knapsack}:
\begin{itemize}
\item First, it does not admit an FPTAS (unless $P=NP$), even if there are only two times steps (Section~\ref{sec:nofptas}) and the bonus is uniform ($B_{ti}=B$ for all $i$, all $t$);
\item Second, the problem is pseudo-polynomial if the number of time steps $T$ is a fixed constant (Section~\ref{sec:pseudo}) but is strongly $NP$-hard in the general case even in the case of uniform bonus (Section~\ref{sec:strong}). 
\end{itemize}


\subsection{No FPTAS}\label{sec:nofptas}

\begin{theorem}\label{thnofptas}
There is no $FPTAS$ for {\sc Multistage Knapsack} unless $P=NP$, even if there are only two time steps and the bonus is uniform.
\end{theorem}
\begin{proof}
We prove the result by a reduction from $\text{CARDINALITY}(2-KP)$, known to be $NP$-complete (Theorem~ \ref{th:dimkp}) 

For an instance $I$ of $\text{CARDINALITY}(2-KP)$, we consider the following instance $I'$ of {\sc Multistage  Knapsack}: \\
\begin{itemize}
\item There are $T=2$ time steps, and the same set of objects $N=\{1,2,...,n\}$ as in $I$.
\item The weights $w_{1i}$ and $w_{2i}$ are the same as in $I$, for all $i \in N$.
\item $p_{1i}=p_{2i}=1$ for all $i \in N$.
\item $B_{1i}=2$ for all $i\in N$.
\end{itemize}

We show that the answer to $I$ is Yes (we can find $K$ objects fulfilling the 2 knapsack constraints) if and only if there is a solution of value at least $2K+2n$ for the instance $I'$ of $MK$.

If we have a set $A$ of $K$ objects for $I$, then we simply take these objects at both time steps in $I'$. This is feasible, the knapsack revenue is $2K$ and the transition revenue is $2n$.

Conversely, consider an optimal solution of value at least $2K+2n$ in $I'$. In an optimal solution, we necessarily have $x_{1j}=x_{2j}$, i.e an object is taken at the time $1$ if and only if it is taken at the time $2$. Indeed, assume that there is one $i \in N$ such that $x_{1i} \neq x_{2i}$, then consider the solution where $x_{1i} = x_{2i}=0$. This is still feasible; the knapsack revenue reduces by 1 but the transition revenue increases by $2$, contradiction.
Thus the same set of objects $A$ is taken at both time steps. Since the value of the solution is at least $2K+2n$, the size of $A$ is at least $K$. Hence the answer to $I$ is Yes.

Note that in $I'$ the optimal value is at most $4n$. We produce in this reduction polynomially bounded instances of {\sc Multistage  Knapsack} (with only two time steps), so the problem does not admit an FPTAS. Indeed, suppose that there is an FPTAS producing a $(1-\epsilon)$-approximation in time $p(1/\epsilon,n)$, for some polynomial $p$. Let $\epsilon = \frac{1}{4n+1}$. If we apply the FPTAS with this value of $\epsilon$ on $I'$ we get a solution $S$ of value at least  $(1-\epsilon)OPT(I')\geq OPT(I') - \frac{OPT(I')}{4n+1} > OPT(I')-1$. Yet all the possible values are integers so $S$ is optimal. The running time is polynomial in $n$, impossible unless $P=NP$.
\end{proof}

\subsection{Pseudo-polynomiality for a constant number of time steps}\label{sec:pseudo}

We show here that the pseudo-polynomiality of the {\sc Knapsack} problem generalizes to {\sc Multistage  Knapsack} when the number of time steps is constant. More precisely, with a standard dynamic programming procedure, we have the following.

\begin{theorem}\label{thpseudo}
{\sc Multistage  Knapsack} is solvable in time $O(T(2C_{max}+2)^Tn)$ where $C_{max}=\max\{C_i,i=1,\dots,T\}$.
\end{theorem}

\begin{proof}

For any $T$-uple $(c_1,\dots,c_T)$ where $0\leq c_i\leq C_i$, and any $s\in \{0,\dots,n\}$, we define $\alpha(c_1,\dots,c_T,s)$ to be the best value of a solution $S=(S_1,\dots,S_T)$ such that:
\begin{itemize}
	\item The weight of knapsack at time $t$ is at most $c_t$: for any $t$, $\sum_{i\in S_t}w_{ti}\leq c_i$;
	\item The solution uses only objects among the first $s$: for any $t$, $S_t\subseteq \{1,\dots,s\}$.
\end{itemize}
The  optimal value of {\sc Multistage  Knapsack} is then $\alpha(C_1,\dots,C_T,n)$. We compute $\alpha$ by increasing values of $s$. For $s=0$, we cannot take any object so $\alpha(c_1,\dots,c_T,0)=0$. 

Take now $s\geq 1$. To compute $\alpha(c_1,\dots,c_T,s)$, we simply consider all the $2^T$ possibilities for taking or not object $s$ in the $T$ time steps. Let $A\subseteq \{1,\dots,T\}$ be a subset of time steps. If we take object $s$ at time steps in $A$ (and only there), we first check if $A$ is a {\it valid} choice, i.e., $w_{ts}\leq c_t$ for any $t\in A$; then we can compute in $O(T)$ the corresponding reward $r_s(A)$ ($\sum_{t\in A}p_{ts}$ plus the transition bonus). We have:

$$\alpha(c_1,\dots,c_T,s)=$$
$$\max\{r_s(A)+\alpha(c_1-w_{1s},\dots,c_T-w_{Ts},s-1):A\subseteq\{1,\dots,T\} \mbox{ valid}\}$$

The running time to compute one value of $\alpha$ is $O(T2^T)$. There are $O(n\Pi_{t=1}^T(C_i+1))=O(n(C_{max}+1)^T)$ values to compute, so the running time follows. A standard backward procedure allows to recovering the solution.
\end{proof}


\subsection{Strongly $NP$-hardness}\label{sec:strong}

\begin{defn}
{\sc Binary Multistage  Knapsack}  is the sub-problem of the Multistage {\sc Knapsack} where all the weights, profits and capacities are all equal to 0 or $1$.
\end{defn}

For the usual {\sc Knapsack} problem, the binary case corresponds to a trivial problem. For the multistage case, we have the following: 

\begin{theorem}
{\sc Binary Multistage  Knapsack} is $NP$-hard, even in the case of uniform bonus.
\end{theorem}

\begin{proof}

We prove the result by a reduction from the {\sc Independent Set} problem where, given a graph $G$ and an integer $K$, we are asked if there exists a subset of $K$ pairwise non adjacent vertices (called an independent set). This problem is well known to be $NP$-hard, see \cite{gj}.

Let $(G,K)$ be an instance of the  {\sc Independent Set} problem, with $G=(V,E)$, $V=\{v_1,\dots,v_n\}$ and $E=\{e_1,\dots,e_m\}$. We build the following instance $I'$ of {\sc Binary Multistage  Knapsack}: 
\begin{itemize}
\item There are $n$ objects $\{1,2\dots,n\}$, one object per vertex;
\item There are $T=m$ time steps: each edge $(v_i,v_j)$ in $E$ corresponds to one time step;
\item at the time step corresponding to edge $(v_i,v_j)$: objects $i$ and $j$ have weight 1, while the others have weight 0, all objects have profit 1, and the capacity constraint is 1.
\item The transition profit is $b_{ti}=B=2nm$ for all $i,t$.
\end{itemize}

We claim that there is an independent set of size (at least) $K$ if and only if there is a solution for {\sc Binary Multistage  Knapsack} of value (at least) $n(m-1)B+mK$. 

Suppose first that there is an independent set $V'$ of size at least $K$. We take the $K$ objects corresponding to $V'$ at all time steps. This is feasible since we cannot take 2 objects corresponding to one edge. The built solution sequence has knapsack profit $mK$ and transition profit $n(m-1)B$ (no modification). 

Conversely, take a solution of {\sc Binary Multistage  Knapsack} of value at least $n(m-1)B+mK$. Since $B=2nm$, there must be no modification of the knapsack over the time. Indeed, otherwise, the transition profit would be at most $n(m-1)B-B$, while the knapsack profit is at most $mn$, so the value would be less than $n(m-1)B$. So we consider the set of vertices corresponding to this knapsack. Thanks to the value of the knapsack, it has size at least $K$. Thanks to the capacity constraints of the bags, this is an independent set. 
\end{proof}
Since $B$ is polynomially bounded in the proof, this shows that {\sc Multistage  Knapsack} is strongly $NP$-hard.
\section{Conclusion}
We considered the {\sc Multistage Knapsack} problem in the offline setting and we studied the impact of the number of time steps in the complexity of the problem. It would be interesting to continue this work in the online and the semi-online settings in order to measure the impact of the knowledge of the future on the quality of the solution.




\bibliography{biblio}


\end{document}